\documentclass[journal]{IEEEtran}
\makeatletter
\def\ps@headings{
\def\@oddhead{\mbox{}\scriptsize\rightmark \hfil \thepage}%
\def\@evenhead{\scriptsize\thepage \hfil \leftmark\mbox{}}%
\def\@oddfoot{}%
\def\@evenfoot{}}
\makeatother
\usepackage{etoolbox}

\usepackage{amsmath, amssymb, bm, cite, epsfig, psfrag}
\usepackage{epstopdf}
\usepackage{graphicx}
\usepackage{algorithm,algpseudocode}
\usepackage[margin=10pt,font=small,labelfont=bf]{caption}
\usepackage{subcaption}
\usepackage{enumerate}
\usepackage{bbold}
\usepackage{tabularx}
\usepackage{multirow}

\newtoggle{conference}
\togglefalse{conference} 

\graphicspath{{figures/}}

\def\beq{\begin{equation}}
\def\eeq{\end{equation}}
\def\beqa{\begin{eqnarray}}
\def\eeqa{\end{eqnarray}}
\def\beqan{\begin{eqnarray*}}
\def\eeqan{\end{eqnarray*}}

\def\R{{\mathbb{R}}}

\def\argmax{\mathop{\mathrm{arg\,max}}}

\newtheorem{lemma}{Lemma}

\def\rhat{\widehat{r}}
\def\rbar{\overline{r}}

\def\wbar{\overline{w}}
\def\xhat{\widehat{x}}
\def\xbar{\overline{x}}
\def\zhat{\widehat{z}}
\def\zbar{\overline{z}}

\def\ra{\rightarrow}

\def\tm1{t\! - \! 1}
\def\tp1{t\! + \! 1}

\def\bbf{\mathbf{b}}

\def\rbf{\mathbf{r}}
\def\rbfhat{\widehat{\mathbf{r}}}
\def\rbfbar{\overline{\mathbf{r}}}

\def\xbf{\mathbf{x}}
\def\xbfhat{\widehat{\mathbf{x}}}
\def\xbfbar{\overline{\mathbf{x}}}

\def\zbf{\mathbf{z}}
\def\zbfbar{\overline{\mathbf{z}}}
\def\zbfhat{\widehat{\mathbf{z}}}
\def\Abf{\mathbf{A}}

\def\Ebf{\mathbf{E}}
\def\Gbf{\mathbf{G}}

\def\Lbar{\overline{L}}

\def\mubf{{\boldsymbol \mu}}

\def\thetabf{{\bm{\theta}}}
\def\thetabfhat{{\widehat{\bm{\theta}}}}
\def\thetabfbar{{\overline{\bm{\theta}}}}

\def\NMS{N_{\rm MS}}
\def\NBS{N_{\rm BS}}

\begin{document}
\bibliographystyle{IEEEtran}

\title{Joint Interference and User Association Optimization
in Cellular Wireless Networks}
\iftoggle{conference}
{
    \author{
        \IEEEauthorblockN{Changkyu Kim, Russell Ford, Yanjia Qi,
        Sundeep Rangan\\}
        \IEEEauthorblockA{Polytechnic Institute of New York University,
        Brooklyn, New York\\
        Email: \{ckim02,rford02\}@students.poly.edu, \{yqi,srangan\}@poly.edu
        }
    }
}{
    \author{
        Changkyu Kim,~\IEEEmembership{Student~Member,~IEEE},
        Russell Ford,~\IEEEmembership{Student~Member,~IEEE},
        Yanjia Qi,~\IEEEmembership{Student~Member,~IEEE},
        Sundeep Rangan,~\IEEEmembership{Senior~Member,~IEEE}
        \thanks{This material is based upon work supported by the National Science
        Foundation under Grants No. 1116589 and 1237821.}
        \thanks{C. Kim (email:ckim02@students.poly.edu),
                R. Ford (email:rford02@students.poly.edu),
                Y. Qi (email: yqi@poly.edu)
                S. Ranagn (email: srangan@poly.edu) are with the
                NYU Wireless Center,
                Polytechnic Institute of New York University, Brooklyn, NY.}
    }
}

\maketitle

\begin{abstract}
In cellular wireless networks, user association refers to the problem of
assigning mobile users to base station cells --
a critical, but challenging, problem in many
emerging small cell and heterogeneous networks.
This paper considers a general class of utility maximization
problems for joint optimization of mobile
user associations and bandwidth and power allocations.
The formulation can incorporate a large class of network topologies,
interference models, SNR-to-rate mappings and network constraints.
In addition, the model can applied in carrier aggregation
scenarios where mobiles can be served by multiple cells simultaneously.
While the problem is non-convex, our main contribution shows that the optimization
admits a separable dual decomposition.  This property
enables fast computation of upper bounds on the utility
as well as an efficient, distributed implementation for approximate
local optimization via augmented Lagrangian techniques.
Simulations are presented in heterogeneous networks with mixtures of macro
and picocells.  We demonstrate significant value of the proposed methods
in scenarios with variable backhaul capacity in the femtocell links and in cases
where the user density is sufficiently low that lightly-used cells can reduce power.
\end{abstract}

\iftoggle{conference}{}{
}

\section{Introduction}

In cellular wireless networks, \emph{user association} refers to the problem
of assigning each mobile terminal a serving base station cell.
In traditional cellular systems, user association is relatively
straightforward in that mobiles can simply connect to the cell
with the strongest received signal strength.  This simple selection policy maximizes
the SNR to each user and hence the rate per unit bandwidth.

However, in many emerging cellular network deployment models, user association
requires consideration of other factors
in addition to signal strength.  For example, to scale capacity
in a cost effective manner, traditional macrocellular
networks are being increasingly supplemented with overlays of
smaller micro- and picocells~
\iftoggle{conference}{\cite{QCOM-HetNetSurvey:11}}{
\cite{QualcommHetNet:11,QCOM-HetNetSurvey:11}}.
The resulting \emph{heterogeneous networks} (HetNets) will
consist of cells with vastly different sizes and require load balancing to
encourage mobiles to connect to smaller cells that would otherwise be
under-utilized~\cite{Ismail:11,YeRon:12}.
In addition, these HetNets may also contain
open \emph{femtocells}~\cite{ChaAndG:08,AndrewsCDRC:12}
where the wired connectivity would be provided by third parties
other than the cellular provider.
Third party backhaul capacity may be highly variable
and may also need to be considered in the user association decisions.
Finally, the state-of-the-art cellular standards
offer a number of advanced intercellular interference coordination
(ICIC) mechanisms
including subband scheduling and beamforming~\cite{3GPPICIC}.  As a result, the throughput
a user can experience in any one cell can depend in a complex manner on the particular
choice of user assignments and resource allocation decisions in other cells.

To understand these issues,
this paper considers a general class of problems for joint optimization of
user association and bandwidth and power
allocations for interference coordination in cellular wireless networks.
The joint user association-interference problem is formulated in the classical
framework of utility maximization~\cite{KelleyMT:98,ShakkottaiS:07}.
Specifically, the throughput to the users is regarded as a function of the cells
that they are assigned to as well as the resources (power or bandwidth) allocated
to the users in the cells.  The bandwidth and power allocations in one cell
may affect the interference levels in other cells.
The problem is to match the users to cells and allocate resources within each cell
to control the inter-cellular interference
and maximize some system-wide utility function of the rates.

Two variants of the problem are considered:  (a) \textbf{Multiflow optimization}
where mobiles can be served by multiple cells simultaneously as in the case
of 3GPP LTE with carrier aggregation (CA)~\cite{3GPP36.300,YuanZWY:10}; and
(b) \textbf{user association optimization} where each mobile can only be served
by one base station cell at a time as in the case of standard 3GPP LTE without
advanced CA capability.

Both problems are, in general, non-convex.  However, while we cannot find an
optimal solution, our main contribution shows that
under very general formulations, the multiflow
optimization admits a separable dual decomposition.
This separability property
has several key implications:  First, for any given values of the Lagrange parameters,
the dual problem can be solved easily and hence one can minimize a dual upper bound exactly.
Therefore, we can obtain a computable upper bound against which we can
compare the performance of any practical, suboptimal algorithm.
Secondly, although the problem may have a non-zero duality gap (i.e.\ the problem may not be
strongly dual), the separability of the dual problem enables efficient
implementation of various augmented Lagrangian techniques for constrained maximization.
Finally, dual techniques have the benefits that they readily lead to
distributed implementations where the base stations broadcast the dual parameters
and the mobiles update the cell selections.  In the multiflow
problem, the Lagrange parameters have a natural interpretation as access prices.

For user association, we recognize that the problem
is essentially classic route selection in networking theory. Route selection
is known in general to be NP-hard, but we follow a standard procedure where we
use the solution to the multiflow problem as an upper bound and then
truncate the solution for the single path constraint~\cite{FordFull:58,GomHu:61,PioroMedhi:04}.
In this way, we can apply dual decomposition plus one additional truncation step
for the user association problem.
This mutliflow approximation has
also been used in cellular user association in~\cite{YeRon:12} (see below).

An appealing feature of our approach is that the dual separability
applies under very general circumstances.
Specifically, our formulation incorporates
a general linear mixing interference model,
first presented in \cite{RanganM:12} for ICIC problems in femtocells.
Although we will only study bandwidth and power allocations in this paper, the model
can incorporate -- with some extensions -- other ICIC techniques including
subband allocations and beamforming.
In addition, our only restriction on the utility functions is that they are separable.
The dual decomposition will hold even under non-concave utilities.
Finally, unlike previous formulations, the approach here can incorporate
network constraints (provided they are known) including limits in the backhaul
-- an important feature for third-party offload.

\subsection{Related Work}

Early work in the area of optimal user association includes
the classic paper by Hanly~\cite{Hanly:95} that found
a simple iterative algorithm for jointly optimal user association and
uplink power control in CDMA systems.
\iftoggle{conference}{}{
Downlink versions of this problem were later considered in~\cite{Smolyar:09}.
Although we will focus on
bandwidth allocation for OFDMA systems in this paper, our methods apply to
general linear mixing interference models, and thus power control may also be in
considered in the future.

}The recent interest in user association is however not for power control, but rather
optimized load balancing in heterogeneous networks as described
in \cite{Madan:10,Ismail:11,YeRon:12}.
Solutions in 3GPP standards bodies have focussed on simple ``range expansion"
techniques that apply a fixed
bias in the cell selection procedure to shift mobiles to the
smaller cells~\cite{3GPPR1-083813:08,3GPPR1-094225:09}.
These techniques can be supplemented with adaptive bias based on SINR~\cite{Ismail:11}
and iterative transmit power selection~\cite{Madan:10}.
The cell-site selection problem is addressed as utility-based optimization problems with backhaul capacity constraints in~\cite{Galeana:09},
while \cite{Kim:10} proposes a distributed algorithm based on
an optimal load of base stations.  Although our evaluation methodology is based on
standard simulation models such as~\cite{3GPP36.814},
stochastic analyses have been provided for these range expansion techniques
in \cite{JoSXA:11}.
The methodology in this paper follows most closely to the work~\cite{YeRon:12}
that used a similar optimization technique with a multiflow upper bound and
dual decomposition. \iftoggle{conference}{}{The work~\cite{YeRon:12} however does not consider
joint optimization of user association with power and bandwidth control.}

One of the interesting findings in \cite{Madan:10,Ismail:11,YeRon:12}
was that static range expansion works remarkably well and more sophisticated optimization approaches offer little benefit.
However, here we will see that, when combined with
joint bandwidth / power optimization, or in the
presence of network constraints, optimized
user association can offer larger gains in some scenarios.
For example, as our simulations will indicate,
when the network loading per cell is light or when the backhaul capacity is limited,
unused cells can reduce power improving the interference conditions in other cells.

To perform the joint optimization of user association and interference coordination,
our work synthesizes several well-known network optimization methods.
The use of multipath upper bounds for optimized
route selection is a classic technique in networking
\iftoggle{conference}{\cite{PioroMedhi:04}}{
theory that dates at least back to~\cite{FordFull:58,GomHu:61} -- see the
text~\cite{PioroMedhi:04} for a historical survey as well as some
testable conditions under which the multipath
upper bound is provable optimal.  See also~\cite{Jef:1995,KelleyMT:98,Kar:2001,Lin-Shroff:2006}.}
Joint power and congestion problems is given in the well-known paper~\cite{chiang:05}.
\iftoggle{conference}{}{
See also \cite{HandeRCW:08} and the monograph~\cite{ChiangHLT:08}.
This works also uses dual decomposition techniques, but applies to a potentially
more general linear mixing interference model that includes power control as a
special.  However, while more general, our methodology has no guarantee of
global optimality.

Utility optimization approaches to the multipath
problems have been extensively studied including both primal and dual solutions
with close connections to rate control~\cite{KelleyMT:98} and
network pricing~\cite{Jef:1995}.  The dual decomposition methods used in this paper
are closest to~\cite{Kar:2001,Lin-Shroff:2006} which explore
augmented Lagrangian and method of multiplier techniques.}

\section{General Flow and Interference Control Problem} \label{sec:genModel}

Before describing the user association and multiflow problems in
cellular networks, we first
describe a general joint flow control and interference optimization problem.
We will apply this general formulation to cellular networks in the next section.

For the general problem, we
consider a system with $S$ flows with rates $r_s$, $s=1,\ldots,S$
where a subset of the flows are constrained by interfering wireless links.
Index the wireless links by $\ell = 1,\ldots,L$ and
let $\Gamma(\ell)$ denote the set of flows that share the wireless link $\ell$.
We assume that wireless link capacity places a constraint on the rates
through the link of the form
\beq \label{eq:Cxzell}
    \sum_{s \in \Gamma(\ell)} r_s \leq C_\ell(x_\ell,z_\ell), \quad
    \ell = 1,\ldots,L,
\eeq
where $x_\ell$ denotes a resource allocation variable for the link,
$z_\ell$ is the interference level on that link and $C_\ell(x_\ell, z_\ell)$ is the
wireless link capacity.
Following~\cite{RanganM:12}, our key assumption will be that
the interference level on the links
are a linear function of the resource allocation variables on the other links so that
we can write
\beq \label{eq:zGx}
    \zbf = \Gbf\xbf,
\eeq
for some gain matrix $\Gbf$.
In addition, we assume that there are some linear constraints on the
resource allocation variables of the form
\iftoggle{conference}{$\Abf_x\xbf \leq \bbf_x$}{
\beq \label{eq:Axb}
    \Abf_x\xbf \leq \bbf_x,
\eeq
} for some matrix $\Abf_x$ and vector $\bbf_x$.
In the sequel, it will be convenient to rewrite the $L$
constraints \eqref{eq:Cxzell} in a single matrix form
\beq\label{eq:Cxz}
  \Ebf\rbf \leq C(\xbf,\zbf),
\eeq
where $\Ebf$ is the incidence matrix
\iftoggle{conference}{and $C(\xbf,\zbf)$ is the vector of capacities
$C_\ell(x_\ell,z_\ell)$
}{with components
\[
    E_{\ell s} = \left\{ \begin{array}{cc}
        1 & \mbox{if  } s \in \Gamma(\ell), \\
        0 & \mbox{if  } s \not \in \Gamma(\ell),
        \end{array} \right.
\]
and $C(\xbf,\zbf)$ is the capacity vector
\[
    C(\xbf,\zbf) = \left(C_1(x_1,z_1), \cdots, C_L(x_L,z_L)\right)^T.
\]
}

In the user selection problem below, $x_\ell$ will be a quantity proportional
to the power allocated to a wireless link $\ell$.  Hence the interference powers
$z_\ell$ are naturally described as a linear function of the form
\eqref{eq:zGx} where the matrix $\Gbf$ would depend on the path losses
between the transmitters and receivers.  The link capacity \eqref{eq:Cxzell}
could then be a function of the signal-to-noise ratio $x_\ell/z_\ell$.
As described in~\cite{RanganM:12}, the linear mixing interference model \eqref{eq:zGx}
can also be used for studying more complex interference scenarios including
links with subbands or beamforming.

As discussed in the Introduction, we also wish to incorporate rate constraints in
the wired network.
To this end, we assume the network constraints are of the form \cite{PioroMedhi:04}
\beq \label{eq:Arb}
    \Abf_r\rbf \leq \bbf_r,
\eeq
for some matrix $\Abf_r$ and vector $\bbf_r$.

Our goal is to select a resource allocation vector $\xbf$ and flow rate vector
$\rbf$ to maximize some utility function of the form
\beq
    U(\rbf) := \sum_{s=1}^S U_s(r_s),
\eeq
where $U_s(r_s)$ is the utility of the rate on flow $s$.
\iftoggle{conference}{}{
Our formulation will not require that the utility functions are concave,
only that they are continuous and monotonically increasing in $\rbf$.

\begin{figure}
  \centering
  \includegraphics[width=0.25\textwidth]{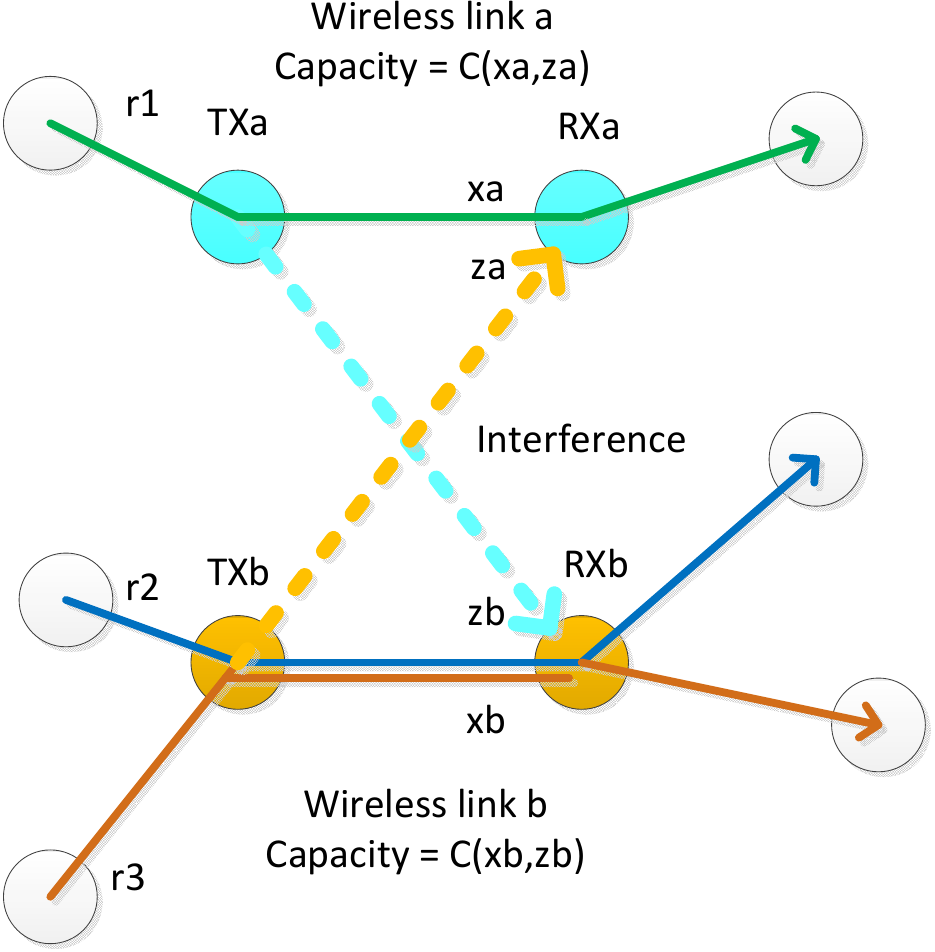}
  \caption{\label{fig:multPathEx} Simple example of a network with
   three flows and two interfering wireless links for joint power
   control and rate optimization. }
\end{figure}

To illustrate the formulation,
Fig.~\ref{fig:multPathEx} shows a simple example for joint power control
and rate optimization.  The network has
$S=3$ flows with rates $r_1$, $r_2$ and $r_3$.  Two of the links,
labeled links $a$ and $b$, are wireless links that interfere with one another.
The transmit signal powers on the links are denoted $x_a$ and $x_b$, while $z_a$
and $z_b$ represent the interference powers.  In this case,
\[
    z_a = G_{ab}x_b, \quad z_b = G_{ba}x_a,
\]
where $G_{ab}$ and $G_{ba}$ are the cross gains between the transmitters
and receivers.  The SINRs on the two links would be given by
\[
    \gamma_a = \frac{G_{aa}x_a}{z_a + n_a}, \quad
    \gamma_b = \frac{G_{bb}x_b}{z_b + n_b},
\]
where $n_a$ and $n_b$ are the noise powers and $G_{aa}$ and $G_{bb}$
are the gains on the intended links.
If the link capacities $C_a$ and $C_b$
can be modeled as functions of the SINRs $\gamma_a$ and $\gamma_b$,
then we can write
\[
    r_1 \leq C_a = C_a(x_a,z_a), \quad r_2+r_3 \leq C_b = C_b(x_b,z_b),
\]
which fits in the format of \eqref{eq:Cxzell}.
Under this model,
one can then attempt to jointly attempt to optimize the transmit powers $x_a$ and $x_b$
along with the flow rates $r_1$, $r_2$ and $r_3$ to maximize some utility.
}

To summarize the optimization, let $\thetabf$ denote the vector of all
the decisions variables
\beq \label{eq:theta}
    \thetabf = (\rbf,\xbf,\zbf).
\eeq
\iftoggle{conference}{}{
and let $\thetabfbar$ denote the maximum values
\beq \label{eq:thetabar}
    \thetabfbar = (\rbfbar,\xbfbar,\zbfbar).
\eeq
}
With some abuse of notation, we let $U(\thetabf) = U(\rbf)$, so that we can
rewrite the optimization as
\beq \label{eq:opt}
    \max_{0 \leq \thetabf \leq \thetabfbar} U(\thetabf) \quad \mbox{s.t. }
        g(\thetabf) \leq 0,
\eeq
where $g(\thetabf)$ is the vector of constraints
\beq \label{eq:gtheta}
    g(\thetabf) := \left[ \begin{array}{c}
        \Abf_x\xbf-\bbf_x \\
        \Abf_r\rbf-\bbf_r \\
        \Gbf\xbf - \zbf \\
        \Ebf \rbf - C(\xbf,\zbf)
        \end{array} \right],
\eeq
and $\thetabfbar$ is some upper bound on the variables.

\section{User Association and Multiflow in Cellular Networks}\label{sec:multiflow}

Having described the general flow and interference control problem, we now turn to
special case of user association and bandwidth optimization in cellular systems.
Our model considers the downlink of a cellular network
with $\NMS$ mobile stations, MS $i$, $i=1,\ldots,\NMS$
and $\NBS$ base stations, BS $j$, $j=1,\ldots,\NBS$.
We suppose that there are a total of $L$ flows, with each flow $\ell$ passing through
exactly one BS-MS pair.
Although a single base station may support multiple flows, we will logically
regard each flow $\ell$ as a separate wireless link with rate $r_\ell$ and
bandwidth allocation $x_\ell$.
For simplicity, we will assume only the BS-MS links are wireless
with interference constraints.  Any other (wired) links are assumed to have
fixed capacity constraints.
We let $\Gamma_{\rm MS}(i)$ denote the set of flows received by MS~$i$ and
$\Gamma_{\rm BS}(j)$ be the set of flows transmitted by BS~$j$.

\begin{figure}
  \centering
  \includegraphics[width=0.4\textwidth]{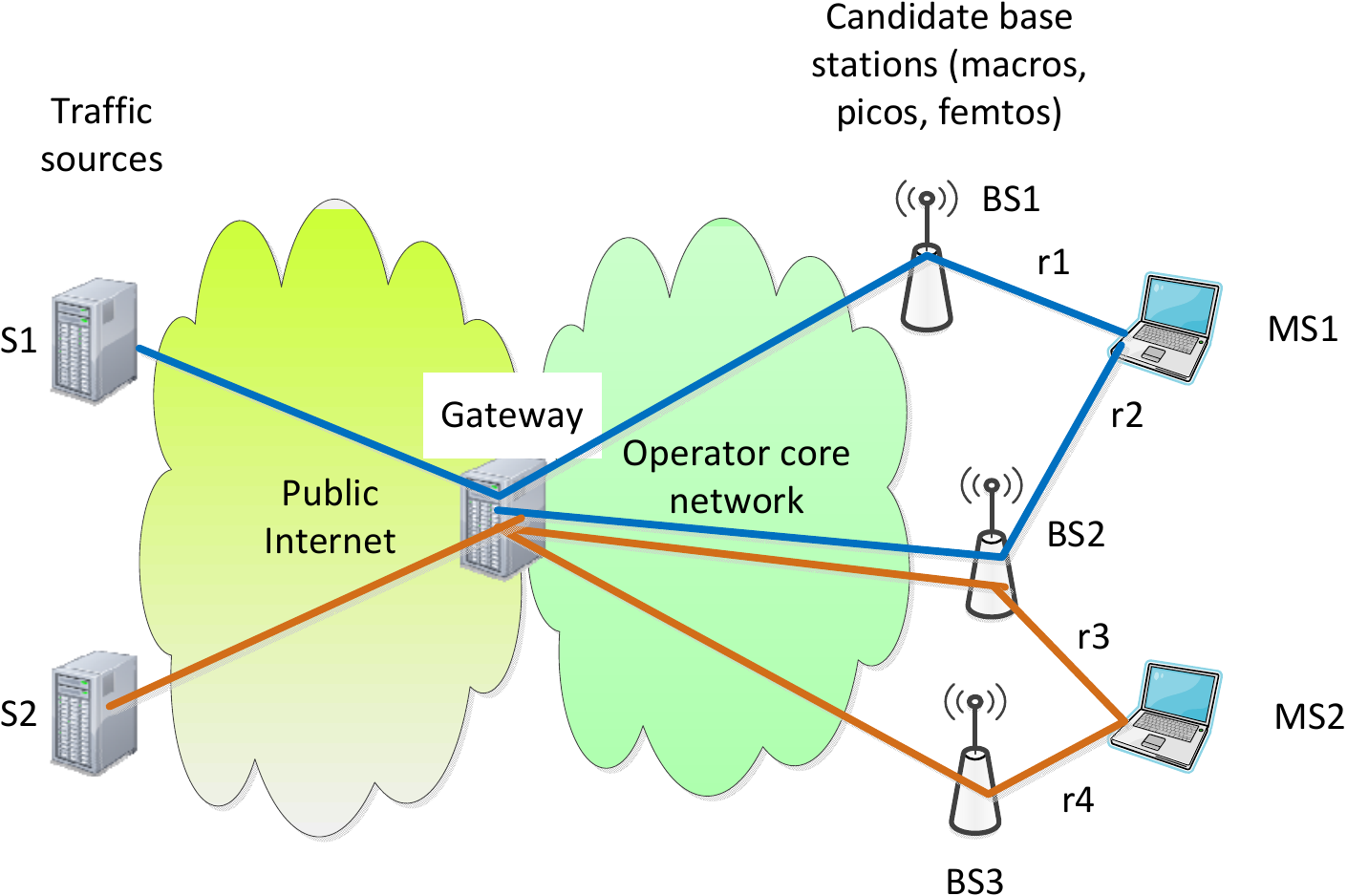}
  \caption{\label{fig:cellSelNetEx} Simple example network where two
  mobiles, MS1 and MS2, must each select between two base stations.
  Since each of the two MS's have two candidate
flows, there are a total of $L = 4$ flows with sets
$\Gamma_{MS}(1) = \{1,2\}$, $\Gamma_{MS}(2) = \{3,4\}$,
$\Gamma_{BS}(1) = \{1\}$, $\Gamma_{BS}(2) = \{2,3\}$ and
$\Gamma_{BS}(3) = \{4\}$. }
\end{figure}

To illustrate the model and notation,
Fig.~\ref{fig:cellSelNetEx} shows a simple example of a network
with two mobile stations, MS1 and MS2, in a typical cellular configuration.
Mobiles MS1 and MS2 received data from sources S1 and S2 in the public Internet.
Following a standard cellular architecture such as 3GPP LTE \cite{OlssonSRFM:09},
the data is assumed to arrive from the public Internet to a common gateway
which then tunnels the traffic through an operator-controlled core network
to the base station selected to the serve the mobiles.  In the simple
example in Fig.~\ref{fig:cellSelNetEx} we have three candidate base stations
shared amongst the two mobiles.
Since each of the two MS's have two candidate
BSs, there are a total of $L = 4$ flows, each flow following a path of the form:
source $\ra$ gateway $\ra$ BS $\ra$ MS.

Now returning to the general case,
let $\rbar_i$ be the total rate to MS$i$, which must satisfy the constraint
\beq \label{eq:rbarMF}
    \rbar_i \leq \sum_{\ell \in \Gamma_{\rm MS}(i)} r_\ell, \quad
    i = 1,\ldots,\NMS.
\eeq
We let $\rbf$ be the vector containing both the total  and individual flow rates
\beq \label{eq:rvecMF}
    \rbf = (\rbar_1,\ldots,\rbar_{\NMS}, r_1,\ldots,r_L)^T.
\eeq
We assume the utility is a separable function of the total rates to each MS
so that
\beq \label{eq:utilMF}
    U(\rbf) = \sum_{i=1}^{\NMS} U_i(\rbar_i),
\eeq
for some utility functions $U_i(\rbar_i)$.

To model the bandwidth constraints and interference,
we assume that each BS $j$ has a total available bandwidth $\wbar_j$ and
radiates a fixed power per unit bandwidth $P_j$.
We let $x_\ell$ be the bandwidth allocated to the flow on link $\ell$.
The bandwidth allocations at BS $j$ must thus satisfy the constraint
\beq \label{eq:bwsum}
    \sum_{\ell \in \Gamma_{\rm BS}(j)} x_\ell \leq \wbar_j, \quad
    j = 1,\ldots,\NBS.
\eeq
Also, the total power radiated by BS $j$ is given by
\[
    P_j\sum_{\ell \in \Gamma_{\rm BS}(j)} x_\ell.
\]
Now suppose that a flow $\ell$ is served by BS $j$ and received by MS $i$.  Then,
the total interference power received on that flow is given by
\beq \label{eq:zpowMF}
    z_\ell = \sum_{k \neq j} H_{\ell k}P_k\sum_{\ell \in \Gamma_{\rm BS}(k)} x_\ell,
\eeq
where $H_{\ell k}$ is the channel gain from BS$k$ to the receiver for flow $\ell$.
Since flow $\ell$ is served by BS$j$, the SINR on the link is then
\[
    \gamma_\ell(z_\ell) = \frac{H_{\ell j}P_j}{z_\ell + N_\ell},
\]
where $N_\ell$ is the thermal noise at the receiver.
We can then assume the capacity constraint on the $\ell$-th link is of the form
\beq \label{eq:capMF}
    r_\ell \leq C(x_\ell,z_\ell) := x_\ell \rho_\ell\bigl(\gamma_\ell(z_\ell)\bigr),
\eeq
where $\rho_\ell$ is the spectral efficiency (rate per unit bandwidth) on the
$\ell$-th link as a function of the SINR.  The spectral efficiency $\rho_\ell$
is multiplied by the allocated bandwidth $x_\ell$ to yield the link capacity
$C_\ell(x_\ell,z_\ell)$ in \eqref{eq:capMF}.  Our formulation will permit any
continuous spectral efficiency function.

Under this model, we consider two problems:
\begin{itemize}
\item \textbf{Multiflow optimization:}  In this case,
there are no additional restrictions on the rate vector $\rbf$ so that
the mobile can be served from multiple base station cells simultaneously.
In 3GPP LTE-Advanced, this scenario would correspond to \emph{carrier aggregation}~
\cite{3GPP36.300,YuanZWY:10}.

\item \textbf{User association optimization:}  In this case,
we assume each mobile can be served by only one cell at a time,
which is the standard model for LTE without the advanced carrier aggregation feature.
This requirement can be modeled via an
additional constraint:  for all $i=1,\ldots,\NMS$,
\beq \label{eq:singPathCon}
    r_\ell = 0 \mbox{ for all but one } \ell \in \Gamma_{\rm MS}(i).
\eeq
Under this constraint, the selection of the index $\ell$ such that $r_\ell \neq 0$,
determines which BS will serve  MS~$i$.
\end{itemize}

Now, the multiflow optimization is exactly
a special case of the general flow and interference control problem in
Section~\ref{sec:genModel}.
Specifically, let $\xbf \in \R^L$ be
the vector of bandwidth allocations and $\zbf \in \R^L$ be the vector of interference
powers.  The bandwidth constraints \eqref{eq:bwsum} can be written of the form
$\Abf_x\xbf \leq \bbf_x$ for some $\Abf_x$ and $\bbf_x$.  The
network constraints will be written in the form $\Abf_r\rbf \leq \bbf_r$,
which can include constraints on either the flow rates or total rates.
The interference powers \eqref{eq:zpowMF} can be written as $\zbf = \Gbf\xbf$
for some gain matrix $\Gbf$ and the capacity \eqref{eq:capMF} can be written
in the form \eqref{eq:Cxz}.

On the other hand, the single path constraint \eqref{eq:singPathCon}
for the user association optimization problem cannot be directly placed into
the form of flow and interference control problem in Section~\ref{sec:genModel}.
As we will discuss in the next section, we will use the multiflow optimization
as an upper bound for this problem
and then truncate the output to obtain a feasible single path
solution.

\section{Approximate Optimization via Dual Decomposition}
\label{sec:dual}

\subsection{General Flow and Interference Control} \label{sec:genOpt}
We now describe potential methods for solving the optimization problems.
We begin with the general flow and interference control
problem in Section~\ref{sec:genModel}.
The maximization of \eqref{eq:opt} in this problem is non-convex
in general since the
capacity function $C(\xbf,\zbf)$ may have an arbitrary form.  In addition,
the utility function $U(\rbf)$ may be non-concave.
However, since both the objective functions and constraints
in \eqref{eq:opt} admit a separable structure, it is natural
to attempt to solve the \eqref{eq:opt} via dual decomposition.

To this end, corresponding to the constrained optimization \eqref{eq:opt},
define the Lagrangian \cite{NorWright:06}
\beqa
    \lefteqn{ L(\thetabf, \mubf) := U(\thetabf) - \mubf^Tg(\thetabf) } \nonumber \\
    &=& U(\rbf) + \mubf^{xT}(\bbf_x-\Abf_x\xbf) + \mubf^{rT}(\bbf_r-\Abf_r\rbf)
        \nonumber \\
    &+&  \mubf^{zT}(\zbf-\Gbf\xbf) + \mubf^{cT}(C(\xbf,\zbf) - \Ebf\rbf),
    \hspace{2cm}
    \label{eq:Lag}
\eeqa
where $\mubf \geq 0$ are the dual parameters which are partitioned conformably
with $g(\thetabf)$ in \eqref{eq:gtheta} as
\beq \label{eq:mu}
    \mubf = (\mubf^x, \mubf^r, \mubf^z, \mubf^c)^T.
\eeq
The following simple lemma provides our main motivation for considering
duality-based optimizations:
namely that the dual maximization is separable:

\begin{lemma} \label{lem:LagSep}
Let $\Phi(\thetabf)$ be any separable function of the form
\beq \label{eq:barrier}
    \Phi(\thetabf) = \Phi(\rbf,\xbf,\zbf) := \sum_i \phi^r_i(r_i) +
    \sum_\ell \phi^x_{\ell}(x_\ell,z_\ell),
\eeq
for some functions $\phi^r_i(\cdot)$ and $\phi^x_\ell(\cdot)$.
Then, the maximization of the Lagrangian $L(\thetabf,\mubf)$ in
\eqref{eq:Lag} \emph{augmented} by $\Phi(\thetabf)$,
\beq \label{eq:thetahat}
    \thetabfhat(\mubf) = (\rbfhat(\mubf), \xbfhat(\mubf), \zbfhat(\mubf))
        := \argmax_{\thetabf \leq \thetabfbar} \Bigl[ L(\thetabf,\mubf)
        -  \Phi(\thetabf) \Bigr]
\eeq
has components given by the solutions to the optimizations
\beqa
    \lefteqn{ \rhat_s \in
        \argmax_{r_s} \Bigl[ U_s(r_s) - \phi^r_s(r_s)
        - (\Ebf^T\mubf^c+\Abf_r^T\mubf^r)_s r_s \Bigr] }  \nonumber \\
        & & \hspace{0.5cm} \mbox{s.t.} \quad 0 \leq r_s \leq \rbar_s
        \hspace{4cm} \label{eq:rhatmu}
\eeqa
and
\beqa
    \lefteqn{ (\xhat_\ell,\zhat_\ell) \in
        \argmax_{x_\ell,z_\ell} \Bigl[ \mu^c_\ell C_\ell(x_\ell,z_\ell) -
             \phi^x_\ell(x_\ell) -  \phi^z_\ell(z_\ell) } \nonumber \\
        & & \hspace{2cm} - (\Abf_x^T\mubf^x + \Gbf^T\mubf^z)_\ell x_\ell
        + \mu^z_\ell z_\ell \Bigr]  \hspace{1cm} \nonumber \\
        & & \hspace{1cm} \mbox{s.t.} \quad 0 \leq x_\ell \leq \xbar_\ell, \quad
        0 \leq z_\ell \leq \zbar_\ell.\label{eq:xzhatmu}
\eeqa
\end{lemma}
\begin{proof}
This result follows immediately from the
separable structure of the objective function and constraints.
\end{proof}

Thus, the vector-valued optimization \eqref{eq:thetahat} decomposes into $S$
one-dimensional optimizations over the variables $r_s$ and $L$ two-dimensional
optimizations over the variables $(x_\ell,z_\ell)$.  In many cases, we will
see that these have explicit closed-form solutions or can be easily computed via
a simple line search.

\iftoggle{conference}{
There are two key consequences to this separability result that follow immediately
from standard optimization theory \cite{NorWright:06}:
First, one can use weak duality to efficiently to compute an
upper bound on the maximum net utility.  The dual upper bound may not
be tight (i.e.\ the optimization may not be strongly dual).
However, a second consequence of a computable dual maxima is that
one can efficiently implement several well-known augmented Lagrangian techniques.
These methods include various inexact versions of alternating direction method
of multiplier methods \cite{BoydPCPE:09,Zhang:JSC:11}
-- see the full paper \cite{KimFQR:13-arxiv} for details.
}{
There are two key consequences to this separability result:

\paragraph{Computation of an Upper Bound}  First, the well-known weak duality
property \cite{NorWright:06} states that an upper bound to the maximum of the
utility is given by
\beq \label{eq:dualBnd}
    \max_{\thetabf~:~g(\thetabf)\leq 0} U(\thetabf) \leq
        \min_{\mubf \geq 0} \Lbar(\mubf),
\eeq
where $\Lbar(\mubf)$ is the dual maxima,
\beq \label{eq:Lbar}
    \Lbar(\mubf) := \max_{\thetabf \leq \thetabfbar} L(\thetabf,\mubf).
\eeq
This bound applies even when $U(\thetabf)$ is non-convex.  Moreover,
$\Lbar(\mubf)$ is also always convex in $\mubf$ with subgradient $g(\thetabf(\mubf))$.
Since Lemma~\ref{lem:LagSep} shows that the maxima
$\thetabfhat(\mubf)$ and $\Lbar(\mubf)$
can be easily evaluated for each $\mubf$, the upper bound in
\eqref{eq:dualBnd} can be minimized via convex programming, thereby
providing a computationally tractable upper bound on the utility.
Of course, since $U(\thetabf)$ and/or the constraints may be non-convex,
the bound may not be tight (i.e.\ the problem is not strongly dual).
Also, in absence of the augmenting term, the subgradient may be non-smooth,
meaning the convex optimization may have a slower (although still feasible)
convergence rate~\cite{Shor:85,Nesterov:05}.

\paragraph{Augmented Lagrangian Optimization}
Secondly, since the maximization \eqref{eq:thetahat} can be computed for any
$\mubf$ and separable augmenting function $\Phi$, one can employ a wide range of
augmented Lagrangian optimization methods for finding local optima for
the constrained maxima problem \eqref{eq:opt}.  Augmented Lagrangian methods
\cite{NorWright:06} generate
a sequence of estimates $\thetabf^t$ typically through the iterations of the form
\begin{subequations} \label{eq:ALopt}
\beqa
    \thetabf^{\tp1} &=& \argmax_{\thetabf \leq \thetabfbar}
        \Bigl[ L(\thetabf, \mubf^t) - \Phi(\thetabf, \thetabf^t) \Bigr]
        \label{eq:thetaAL} \\
    \mubf^{\tp1} &=& \mubf^t - \alpha g(\thetabf^{\tp1}), \label{eq:muAL}
\eeqa
\end{subequations}
where $\alpha > 0$ is a step-size and
$\Phi(\thetabf,\thetabf^t) \geq 0$ is a so-called augmenting function that smooths
the optimization.  Separable forms of the augmenting function can include the
quadratic  $\Phi(\thetabf,\thetabf^t) = \epsilon\|\thetabf-\thetabf^t\|^2_2$
for some $\epsilon > 0$.

}
\subsection{Multiflow Optimization} \label{sec:multiFlowOpt}

Since the multiflow cellular optimization in Section~\ref{sec:multiflow}
is a special case of the general flow and interference control
problem, we can directly apply
the results above in Section~\ref{sec:genOpt} to compute upper bounds
on the utility as well as efficiently implement augmented Lagrangian techniques.
\iftoggle{conference}{In addition, as discussed in the full paper
\cite{KimFQR:13-arxiv}, the augmented Lagrangian methods admit a simple
distributed implementation -- a feature of \cite{YeRon:12}
as well as many duality-based power control algorithms (see the monograph
\cite{ChiangHLT:08}).}{

In addition, in the cellular context augmented Lagrangian iterations of the form
\eqref{eq:ALopt} also admit a simple distributed implementation -- a feature
also found in the similar user association optimization algorithm in \cite{YeRon:12},
as well as many duality-based power control algorithms (see the monograph
\cite{ChiangHLT:08}).
Specifically, it can be verified that, using the separable structure of
the optimization in Lemma~\ref{lem:LagSep}, the rates $r_\ell$ and
bandwidth allocations $x_\ell$ and target interference level $z_\ell$ for each
link can be computed at the base station serving that flow. The base station
only needs local information on the gain matrix $\Gbf$ and dual parameters.
The base stations can then update the dual parameters $\mubf$
as in \eqref{eq:muAL} and broadcast the values to the neighboring base stations
potentially through the mobiles.

}

\subsection{User Association Optimization} \label{sec:userAssocOpt}

As discussed above, the single path constraint \eqref{eq:singPathCon} is nonlinear
and cannot be placed into the framework of the general flow and interference
control problem in Section~\ref{sec:genModel}.  However, a standard method
\iftoggle{conference}{\cite{PioroMedhi:04}}{
\cite{FordFull:58,GomHu:61,PioroMedhi:04} } is to solve the multiflow upper bound
and then simply the ``truncate" the solution to remove all but the best path.
In most cases, the multiflow optimization yields solutions that are zero in all but
one path (the book \cite{PioroMedhi:04} provides conditions under which this
property is provably true) and therefore the truncation often
does not significantly alter the optimization result.  A similar truncation method
is used in \cite{YeRon:12}.

\section{Numerical Simulations}
\subsection{Lightly Loading Heterogeneous Network}

To demonstrate the methodology we simulated the algorithm in
a two-tier heterogeneous network with macro- and picocells loosely following the
3GPP evaluation methodology \iftoggle{conference}{in \cite{3GPP36.814}.}{
in \cite{3GPP36.814}, specifically the macro and
outdoor remote radio head/hotzone scenario in Configuration \#1.}
The parameters are shown in Table~\ref{table:SimPara}.
The main difference is that we consider a much greater density of picocells
(10 instead of 4 picos per macrocell), making the number of UEs per cell small.
We will see that under this lightly-loaded scenario, the power optimization
can offer significant benefits.

\begin{table}[!t]\footnotesize
\renewcommand{\arraystretch}{1.3}
\caption{Simulation parameters}
\label{SimParameter}\label{table:SimPara}
\begin{tabularx}{\linewidth}{|p{2.2cm}|X|}
\hline
\bfseries Parameter & \bfseries \centering  Value \tabularnewline
\hline
\hline
Macro layout & 19 hexagonal cell sites with wraparound, 3 cells per site.
Inter-site distance=500 $m$.
    \tabularnewline \hline
Pico and UE distribution & 10 picos and 10 or 25 UEs per macrocell,
uniformly distributed.
    \tabularnewline \hline
TX power & 46 dBm (macro), 30 dBm (pico)
    \tabularnewline \hline
Path loss & macro:  $131.1+42.8\log_{10}(R)$, \newline
            pico:   $145.4+37.5\log_{10}(R)$
    \tabularnewline \hline
Carrier frequency & 2.1 GHz \tabularnewline \hline
Shadowing & 8 dB lognormal std. dev. \tabularnewline \hline
UE noise figure & 7 dB \tabularnewline \hline
System bandwidth & 10 MHz \tabularnewline \hline
SNR to rate & 3 dB from Shannon, max 4.8 bps/Hz \cite{MogEtAl:07}
    \tabularnewline \hline
Antenna Pattern & macro:  3GPP standard model \cite{3GPP36.814} \newline
                  pico:  omnidirectional \tabularnewline \hline
Traffic model & Full buffer \tabularnewline \hline
\end{tabularx}
\end{table}


We focus on the problem of optimized user association (as opposed to multiflow)
where each UE must select exactly
one cell -- either macro or pico --  since this is main concern in
the cellular standards.  To apply the proposed optimization method to the
user association problem,
after the drop, each UE selected three candidate BSs:  the strongest macrocell
and the two strongest picocells.  The algorithm was then run to maximize
a proportional fair utility, $U_i(r) = \log(r)$, which corresponds to
the sum-log rate.  Fig.~\ref{fig:cellSelSim_noBH}
compares the resulting rate distribution
from the optimization algorithm against simpler, but standard,
cell selection policies.
The curve labeled ``optimizer" is the optimization run, but keeping all
the cells transmitting at maximum power.  In this case, the spectral efficiency
is constant and the algorithm reduces to the method of \cite{YeRon:12}.
The curve labeled ``optimizer+PR" adds power reduction where reducing the
bandwidth allocated in a cell causes its power to reduce.
For comparison, the figure also plots the ``macro only" and ``pico only"
where the UEs use only the macro- or picocells. Both of these options
give very poor rates.
The range expansion methods (labeled RE) select the strongest
cell -- pico or macro -- but the picocells are given a fixed bias
in the received signal strength to encourage mobiles to select the smaller cell.
Manual trials with different
bias levels found an an optimal bias of approximately 6 dB.
Fig.~\ref{fig:cellSelSim_noBH}
shows the rate distribution of RE with bias of 0 and
6 dB.

Table~\ref{table:SimResult} shows the numerical values from the simulation.
With the light loading of 10 UEs per macrocell
(the case plotted in Fig.~\ref{fig:cellSelSim_noBH},
we see that the proposed optimization with power reduction
offers an approximately 23\%
increase in cell throughput and 42\% increase in the 5\% cell edge capacity.
While not dramatic, the gains are larger than previous studies
\cite{Ismail:11,YeRon:12} which showed little gain over simple RE with the correct bias.  The reason for the larger gain here is that,
since we are considering a low UE to cell ratio, some picocells can be
chosen to be unused and can reduce their power.
Indeed, if we go to a higher loading of 25 UEs per macrocell
(2.5 UEs per picocell), the gains from power reduction are reduced since most
of the picocell become fully utilized.

\begin{figure}
\centering
\includegraphics[width=3.0in]{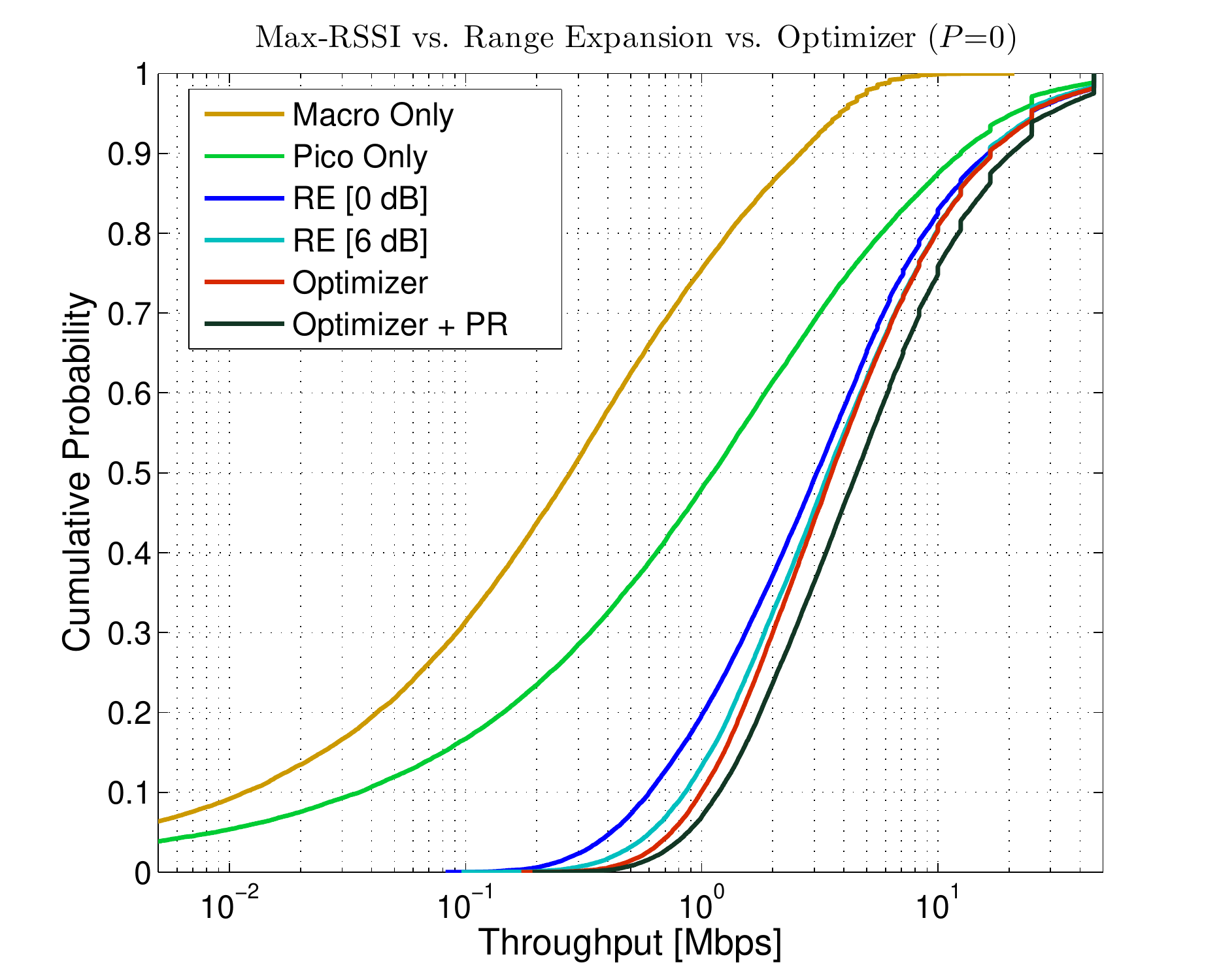}
\caption{Rate distribution comparing various user association algorithms
for 10 UEs per macrocell.  Other parameters given in Table~\ref{table:SimPara}.
}
\label{fig:cellSelSim_noBH}
\end{figure}


\subsection{Backhaul Network Constraints}

In addition to the interference control, a second appealing feature of the
proposed user association methodology is that it can account for backhaul network
constraints.  Such backhaul constraints may arise when open femtocells are used
for coverage.  Femtocells are not connected directly into the operator core network,
but instead use third party backhaul which may be variable in quality.
To simulate this scenario, we employed the same assumptions as before, but
we placed a 1~Mbps limit on 50\% of the picocell links (the picocells now modeling
large numbers of open femtocells).
Fig.~\ref{fig:cellSelSim_BH} shows that the gains from
optimized user association in this scenario are even larger.
Indeed, for 10 UEs per cell, there is an 28\% increase in
average user throughput and 112\% increase in the 5\% cell-edge rate.
In addition, Fig.~\ref{fig:cellSelSim_BH}
shows that RE gives a minimal gains in this heterogeneous network condition
since RE is a power-based selection scheme and does not account for
the backhaul constraints.

\begin{table}\footnotesize
\caption{Simulation Results}
\label{table:SimResult}
\centering
\begin{tabular}{|c|c|c|l|l|}
\hline

\multicolumn{2}{|c|}{\textbf{Method}} & \textbf{RE [6dB]} & \multicolumn{1}{c|}{\textbf{Opt.}} & \multicolumn{1}{c|}{\textbf{Opt. + PR}} \\ \hline

\multirow{2}{*}{\textbf{10UE}} & Sp. Eff.\footnotemark[1] & 0.608 & 0.625 (2.8\%)\footnotemark[3] & 0.748 (23\%)\footnotemark[4] \\ \cline{2-5}
& 5\% Rate\footnotemark[2] & 0.608 & 0.745 (22.6\%) & 0.866 (42.4\%) \\ \hline

\multirow{2}{*}{\textbf{25UE}} & Sp. Eff. & 0.909 & 0.940 (3.4\%) & 0.989 (8.7\%) \\ \cline{2-5}
& 5\% Rate & 0.351 & 0.418 (19.1\%) & 0.435 (23.9\%) \\ \hline

\end{tabular}
\newline
\parbox{0.8\columnwidth}{\footnotesize{1 : Spectral efficiency in bit/sec/Hz \\ 2 : 5\% Cell-edge rate in Mbps \\ 3,4 : gains compared to range expansion [6 dB]}}
\end{table}

\begin{figure}
\centering
\includegraphics[width=3.0in]{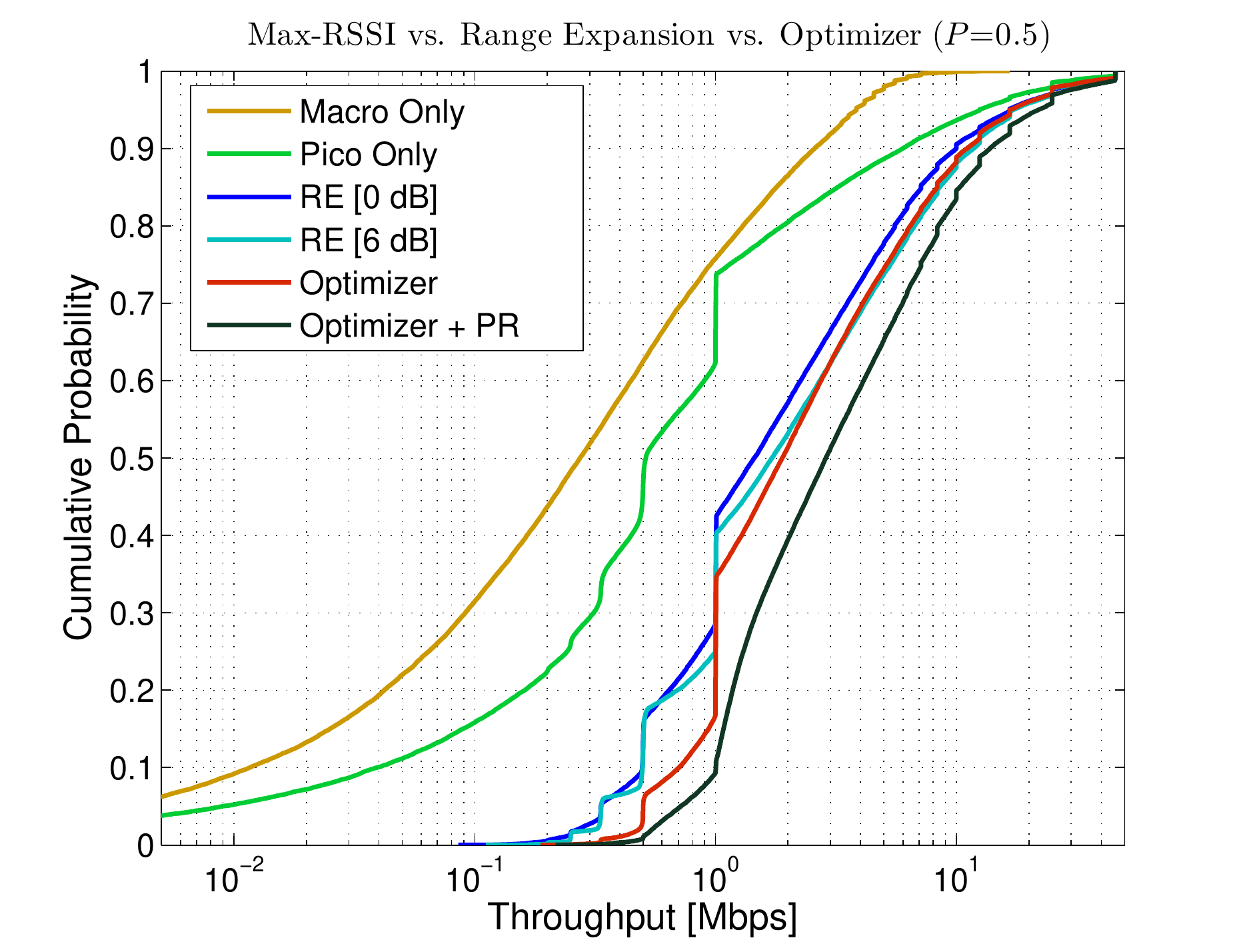}
\caption{Rate distribution of various user association algorithms
with backhaul rate limits on 50\% of the picocells.}
\label{fig:cellSelSim_BH}
\end{figure}

\section{Conclusion}
We have presented a general methodology for joint interference coordination
and user association problems in cellular networks.
The methodology is extremely general in that it can incorporate interference
constraints described by an linear mixing model followed by an arbitrary
interference-to-rate mapping.  Although we have only considered bandwidth allocations
in this paper, based on the results in \cite{RanganM:12},
we believe that we can extend these results to more complex ICIC schemes as well
such as subband allocations and beamforming.  Also, while we cannot find
a provably optimal solution, the dual decomposition
method enables efficient implementation
of suboptimal solutions via augmented Lagrangian techniques as well as
computation of upper bounds on the maximum utility.

\bibliographystyle{IEEEtran}
\bibliography{bibl}
\end{document}